\documentclass[preprint,12pt]{elsarticle}

\usepackage{amssymb}
\usepackage[ruled,vlined,linesnumbered]{algorithm2e}%

\journal{Algorithmica}

\newtheorem{theorem}{Theorem} [section]
\newtheorem{corollary}[theorem]{Corollary}
\newtheorem{lemma}[theorem]{Lemma}
\newenvironment{proof}[1][Proof]{\noindent\textbf{#1.} }{\ \rule{0.5em}{0.5em}}

\begin{document}

\begin{frontmatter}



\title{Generating subgraphs in chordal graphs}

\author[label1]{Vadim E. Levit\corref{cor1}}
\address[label1]{Department of Computer Science, Ariel University, ISRAEL}
\cortext[cor1]{Corresponding Author \\ E-mail addresses: levitv@ariel.ac.il (V. E. Levit), davidt@sce.ac.il (D. Tankus).}

\author[label2]{David Tankus}
\address[label2]{Department of Software Engineering, Sami Shamoon College of Engineering, ISRAEL}


\begin{abstract}
A graph $G$ is \textit{well-covered} if all its maximal independent sets are
of the same cardinality. Assume that a weight function $w$ is defined on its
vertices. Then $G$ is $w$\textit{-well-covered} if all maximal independent
sets are of the same weight. For every graph $G$, the set of weight functions
$w$ such that $G$ is $w$-well-covered is a \textit{vector space}, denoted
$WCW(G)$.

Let $B$ be a complete bipartite induced subgraph of $G$ on vertex sets of
bipartition $B_{X}$ and $B_{Y}$. Then $B$ is \textit{generating} if there
exists an independent set $S$ such that $S \cup B_{X}$ and $S \cup B_{Y}$ are
both maximal independent sets of $G$. In the restricted case that a generating
subgraph $B$ is isomorphic to $K_{1,1}$, the unique edge in $B$ is called a
\textit{relating edge}. Generating subgraphs play an important role in finding
$WCW(G)$.

Deciding whether an input graph $G$ is well-covered is \textbf{co-NP}%
-complete. Hence, finding $WCW(G)$ is \textbf{co-NP}-hard. Deciding whether an
edge is relating is \textbf{NP}-complete. Therefore, deciding whether a
subgraph is generating is \textbf{NP}-complete as well.

A graph is chordal if every induced cycle is a triangle. It is known that
finding $WCW(G)$ can be done polynomially in the restricted case that $G$ is
chordal. Thus recognizing well-covered chordal graphs is a polynomial problem.
We present a polynomial algorithm for recognizing relating edges and
generating subgraphs in chordal graphs.


\end{abstract}


\end{frontmatter}





\section{Introduction}

\subsection{Basic definitions and notation}

Throughout this paper $G$ is a simple (i.e., a finite, undirected, loopless
and without multiple edges) graph with vertex set $V(G)$ and edge set $E(G)$.
Cycles of $k$ vertices are denoted by $C_{k}$. When we say that $G$ does not
contain $C_{k}$ for some $k \geq3$, we mean that $G$ does not admit subgraphs
isomorphic to $C_{k}$. Note that these subgraphs are not necessarily induced.

Let $u$ and $v$ be two vertices in $G$. The \textit{distance} between $u$ and
$v$, denoted $d(u,v)$, is the length of a shortest path between $u$ and $v$,
where the length of a path is the number of its edges. If $S$ is a non-empty
set of vertices, then the \textit{distance} between $u$ and $S$, is defined as
$d(u,S)=\min\{d(u,s):s\in S\}$.

For every positive integer $i$, denote
\[
N_{i}(S)=\{x\in V\left(  G\right)  :d(x,S)=i\},
\]
and
\[
N_{i}\left[  S\right]  =\{x\in V\left(  G\right)  :d(x,S)\leq i\}.
\]

If $S$ contains a single vertex, $v$, then we abbreviate $N_{i}(\{v\})$,
$N_{i}\left[  \{v\}\right]  $ to be $N_{i}(v)$, $N_{i}\left[  v\right]  $,
respectively. We denote by $G[S]$ the subgraph of $G$ induced by $S$. For
every two sets, $S$ and $T$, of vertices of $G$, we say that $S$
\textit{dominates} $T$ if $T\subseteq N_{1}\left[  S\right]  $.

\subsection{Well-covered graphs}

Let $G$ be a graph. A set of vertices $S$ is \textit{independent} if its
elements are pairwise nonadjacent. An independent set of vertices is
\textit{maximal} if it is not a subset of another independent set. An
independent set of vertices is \textit{maximum} if the graph does not contain
an independent set of a higher cardinality.

The graph $G$ is \textit{well-covered} if every maximal independent set is
maximum \cite{plummer:definition}. Assume that a weight function $w:V\left(
G\right)  \longrightarrow\mathbb{R}$ is defined on the vertices of $G$. For
every set $S\subseteq V\left(  G\right)  $, define $w(S)={\sum\limits_{s\in
S}}w(s)$. Then $G$ is $w$\textit{-well-covered} if all maximal independent
sets of $G$ are of the same weight.

The problem of finding a maximum independent set is \textbf{NP-}complete.
However, if the input is restricted to well-covered graphs, then a maximum
independent set can be found in polynomial time using the \textit{greedy
algorithm}. Similarly, if a weight function $w:V\left(  G\right)
\longrightarrow\mathbb{R}$ is defined on the vertices of $G$, and $G$ is
$w$-well-covered, then finding a maximum weight independent set is a
polynomial problem. There is an interesting application, where well-covered
graphs are investigated in the context of distributed $k$-mutual exclusion
algorithms \cite{yaka:coteries}.

The recognition of well-covered graphs is known to be \textbf{co-NP}-complete.
This is proved independently in \cite{cs:note} and \cite{sknryn:compwc}. In
\cite{cst:structures} it is proven that the problem remains \textbf{co-NP}%
-complete even when the input is restricted to $K_{1,4}$-free graphs. However,
the problem can be solved in polynomial time for $K_{1,3}$-free graphs
\cite{tata:wck13f,tata:wck13fn}, for graphs with girth $5$ at least
\cite{fhn:wcg5}, for graphs with a bounded maximal degree \cite{cer:degree},
for chordal graphs \cite{ptv:chordal}, and for graphs without cycles of
lengths $4$ and $5$ \cite{fhn:wc45}.

For every graph $G$, the set of weight functions $w$ for which $G$ is
$w$-well-covered is a vector space \cite{cer:degree}. That vector space is
denoted $WCW(G)$ \cite{bnz:wcc4}. Since recognizing well-covered graphs is
\textbf{co-NP}-complete, finding the vector space $WCW(G)$ of an input graph
$G$ is \textbf{co-NP}-hard. However, finding $WCW(G)$ can be done in
polynomial time when the input is restricted to graphs with a bounded maximal
degree \cite{cer:degree}, to graphs without cycles of lengths $4$, $5$ and $6$
\cite{lt:wwc456}, and to chordal graphs \cite{bn:wcchordal}.

\subsection{Generating subgraphs and relating edges}

Further we make use of the following notions, which have been introduced in
\cite{lt:wc4567}. Let $B$ be an induced complete bipartite subgraph of $G$ on
vertex sets of bipartition $B_{X}$ and $B_{Y}$. Assume that there exists an
independent set $S$ such that each of $S\cup B_{X}$ and $S\cup B_{Y}$ is a
maximal independent set of $G$. Then $B$ is a \textit{generating subgraph} of
$G$, and the set $S$ is a \textit{witness} that $B$ is generating. We observe
that every weight function $w$ such that $G$ is $w$-well-covered must
\textit{satisfy} the restriction $w(B_{X})=w(B_{Y})$.

If the generating subgraph $B$ contains only one edge, say $xy$, it is called
a \textit{relating edge }\cite{bnz:wcc4}. In such a case, the equality
$w(x)=w(y)$ is valid for every weight function $w$ such that $G$ is $w$-well-covered.

Recognizing relating edges is known to be \textbf{NP-}complete \cite{bnz:wcc4}%
, and it remains \textbf{NP-}complete even when the input is restricted to
graphs without cycles of lengths $4$ and $5$ \cite{lt:relatedc4}. Therefore,
recognizing generating subgraphs is also \textbf{NP-}complete when the input
is restricted to graphs without cycles of lengths $4$ and $5$. However,
recognizing relating edges can be done in polynomial time if the input is
restricted to graphs without cycles of lengths $4$ and $6$ \cite{lt:relatedc4}%
, and to graphs without cycles of lengths $5$ and $6$ \cite{lt:wwc456}.

It is also known that recognizing generating subgraphs is a polynomial problem
when the input is restricted to graphs without cycles of lengths $4$, $6$ and
$7$ \cite{lt:wc4567}, to graphs without cycles of lengths $4$, $5$ and $6$
\cite{lt:wwc456}, and to graphs without cycles of lengths $5$, $6$ and $7$
\cite{lt:wwc456}.

\subsection{Chordal graphs}

A graph is \textit{chordal (triangulated)} if its every induced cycle is a
triangle \cite{Berge1967}. Finding a maximum weight independent set in a
chordal graph is a polynomial task \cite{frank:chmwis}. Deciding whether a
chordal graph is well-covered can be done polynomially \cite{ptv:chordal}.
Finding $WCW(G)$ can be completed polynomially if $G$ is chordal
\cite{bn:wcchordal}. We present a polynomial time algorithm, which receives as
input a chordal graph $G$ and an induced complete bipartite subgraph $B$. The
algorithm decides whether $B$ is generating.

\section{Polynomial results for chordal graphs}

\subsection{The vector space $WCW(G)$}

A vertex $x$ in a graph $G$ is \textit{simplicial} if $N_{1}[x]$ is a clique.
A \textit{simplicial clique} is a maximal clique containing a simplicial vertex.

\begin{theorem}
\cite{ptv:chordal} \label{wcchord} Let $G$ be a chordal graph. Then $G$ is
well-covered if and only if every vertex of $G$ belongs to exactly one
simplicial clique.
\end{theorem}

In \cite{bn:wcchordal}, a polynomial characterization of $WCW(G)$ is
presented, for the case that $G$ is chordal. The following definitions and
notation are used.

Let $C(G)$ be the set of all simplicial cliques and $sc(G)=|C(G)|$. Let $C \in
C(G)$ be a simplicial clique. The \textit{associated weighting} function,
denoted $f_{C} : V(G) \longrightarrow\mathbb{R}$, is defined as follows. If
$v\in C$ then $f_{C}(v)=1$, otherwise $f_{C}(v)=0$.

\begin{lemma}
\label{wcwsubspace} \cite{bn:wcchordal} $f_{C} \in WCW(G)$ for every graph
$G$, and for each simplicial clique $C \in C(G)$. Moreover, $\{f_{C}:C \in
C(G)\}$ is an independent set of vectors, and $wcdim(G) \geq sc(G)$.
\end{lemma}

\begin{theorem}
\label{wcwchordal} \cite{bn:wcchordal} Let $G$ be a chordal graph. Then
$wcdim(G) = sc(G)$.
\end{theorem}

By Lemma \ref{wcwsubspace}, for every graph $G$, the vector space spanned by
$\{f_{C}:C\in C(G)\}$ is a subspace of $WCW(G)$. Moreover, if $G$ is chordal
then, by Theorem \ref{wcwchordal}, the vector space spanned by $\{f_{C}:C\in
C(G)\}$ coincides with $WCW(G)$. Let $T$ be the set of all simplicial vertices
in a chordal graph $G$, and let $S$ be a maximal independent set of $G[T]$.
Clearly, $C(G)=\{N_{1}(v):v\in S\}$. In order to construct a function $w\in
WCW(G)$, the following algorithm can be implemented. For every $s\in S$,
define $w(s)$ arbitrarily, while for each vertex $v\in V(G)\setminus S$, let
$w(v)=w(N_{1}(v)\cap S)$. In other words, assigning $1$ to one vertex in $S$
and $0$ to all others, we describe a basis of $WCW(G)$. Clearly, this
procedure is polynomial.

\subsection{Generating subgraphs}

The main result of this subsection is a polynomial time algorithm for
recognizing generating subgraphs in chordal graphs. Let $G$ be a chordal
graph, and let $B$ be an induced complete bipartite subgraph of $G$ on vertex
sets of the bipartition $B_{X}=\{x_{1},...,x_{l}\}$ and $B_{Y}=\{y_{1}%
,...,y_{k}\}$, where $l\leq k$.

\begin{lemma}
\label{k1j} $l=1$.
\end{lemma}

\begin{proof}
If $l \geq2$ then $G[\{x_{1},x_{2},y_{1},y_{2}\}]$ is isomorphic to $K_{2,2} =
C_{4}$, which contradicts the fact that $G$ is chordal.
\end{proof}

Since $l=1$, we denote $B_{X}=\{x\}$. For each $V\in\{X,Y\}$ let $S \subseteq B_{V}$ and
$U\in\{X,Y\}\setminus\{V\}$. Define $M_{1}(S)=N_{1}(S) \cap N_{2}(B_{U})$ and $M_{2}(S)=N_{1}(M_{1}(S)) \cap N_{2}(B_{V})$. If $S$ contains a single vertex, $v$, abbreviate $M_{1}(\{v\})$ and $M_{2}(\{v\})$ to $M_{1}(v)$ 
and $M_{2}(v)$, respectively. Define $M_{1}(B)=M_{1}(B_{X}) \cup M_{1}(B_{Y})$ and $M_{2}(B) = M_{2}(B_{X}) \cup M_{2}(B_{Y})$.

\begin{lemma}
\label{mbx} $M_{1}(B_{X})$ and $M_{1}(B_{Y})$ are disjoint and nonadjacent.
\end{lemma}

\begin{proof}
The fact that $M_{1}(B_{X})$ and $M_{1}(B_{Y})$ are disjoint follows
immediately from the definition of $M_{1}(B_{X})$ and $M_{1}(B_{Y})$.

Let $x^{\prime}\in M_{1}(B_{X})$ and $y^{\prime}\in M_{1}(B_{Y})$, and assume
on the contrary that $x^{\prime}\in N(y^{\prime})$. There exists $1\leq i\leq
k$ such that $y^{\prime}\in N_{1}(y_{i})$. Therefore, $C=(x,y_{i},y^{\prime
},x^{\prime})$ is a copy of $C_{4}$. Obviously, $xy^{\prime}\not \in E(G)$ and
$y_{i}x^{\prime}\not \in E(G)$. Hence, $C$ is an induced $C_{4}$, which is a contradiction.
\end{proof}

\begin{lemma}
\label{m2bx} $M_{2}(B_{X})$ and $M_{2}(B_{Y})$ are disjoint and nonadjacent.
\end{lemma}

\begin{proof}
Assume on the contrary that there exists a vertex $v\in M_{2}(B_{X})\cap
M_{2}(B_{Y})$. There exist two vertices $x^{\prime}\in M_{1}(B_{X})\cap
N_{1}(v)$ and $y^{\prime}\in M_{1}(B_{Y})\cap N_{1}(v)$. By Lemma \ref{mbx},
$x^{\prime}$ and $y^{\prime}$ are distinct and nonadjacent. There exists $y\in
B_{Y}$ such that $G[\{x,y,y^{\prime},v,x^{\prime}\}]$ is isomorphic to $C_{5}%
$, which is a contradiction.

Assume, on the contrary, that there exist two adjacent vertices,
$x^{\prime\prime}\in M_{2}(B_{X})$ and $y^{\prime\prime}\in M_{2}(B_{Y})$.
There exist $x^{\prime}\in M_{1}(B_{X})\cap N_{1}(x^{\prime\prime})$ and
$y^{\prime}\in M_{1}(B_{Y})\cap N_{1}(y^{\prime\prime})$. By Lemma \ref{mbx},
$x^{\prime}$ and $y^{\prime}$ are distinct and nonadjacent. There exists $y\in
B_{Y}$ such that $G[\{x,y,y^{\prime},y^{\prime\prime},x^{\prime\prime
},x^{\prime}\}]$ is isomorphic to $C_{6}$, which is a contradiction.
\end{proof}

\begin{lemma}
\label{myi} Let $1\leq i<j\leq k$. Then $M_{1}(y_{i})$ and $M_{1}(y_{j})$ are
disjoint and nonadjacent.
\end{lemma}

\begin{proof}
If there existed a vertex $v\in M_{1}(y_{i})\cap M_{1}(y_{j})$ then
$G[\{x,y_{i},v,y_{j}\}]$ was isomorphic to $C_{4}$.

If $y_{i}^{\prime}\in M_{1}(y_{i})$ and $y_{j}^{\prime}\in M_{1}(y_{j})$ were
adjacent then $G[\{x,y_{i},y_{i}^{\prime},y_{j}^{\prime},y_{j}\}]$ was
isomorphic to $C_{5}$.
\end{proof}

\begin{lemma}
\label{m2yi} Let $1 \leq i < j \leq k$. Then $M_{2}(y_{i})$ and $M_{2}(y_{j})$
are disjoint and nonadjacent.
\end{lemma}

\begin{proof}
Assume, on the contrary, that there exists a vertex $v\in M_{2}(y_{i})\cap
M_{2}(y_{j})$. There exist two vertices, $y_{i}^{\prime}\in M_{1}(y_{i})\cap
N_{1}(v)$ and $y_{j}^{\prime}\in M_{1}(y_{j})\cap N_{1}(v)$. By Lemma
\ref{myi}, $y_{i}^{\prime}$ and $y_{j}^{\prime}$ are distinct and nonadjacent.
Hence, $G[\{x,y_{i},y_{i}^{\prime},v,y_{j}^{\prime},y_{j}\}]$ is isomorphic to
$C_{6}$, which is a contradiction.

Assume, on the contrary, that $y_{i}^{\prime\prime}\in M_{2}(y_{i})$ and
$y_{j}^{\prime\prime}\in M_{2}(y_{j})$ are adjacent. There exist two vertices,
$y_{i}^{\prime}\in M_{1}(y_{i})\cap N_{1}(y_{i}^{\prime\prime})$ and
$y_{j}^{\prime}\in M_{1}(y_{j})\cap N_{1}(y_{j}^{\prime\prime})$. By Lemma
\ref{myi}, $y_{i}^{\prime}$ and $y_{j}^{\prime}$ are distinct and nonadjacent.
Hence, $G[\{x,y_{i},y_{i}^{\prime},y_{i}^{\prime\prime},y_{j}^{\prime\prime
},y_{j}^{\prime},y_{j}\}]$ is isomorphic to $C_{7}$, which is a contradiction.
\end{proof}

Define a function $f:2^{M_{2}(B)}\longrightarrow2^{M_{1}(B)}$ as $f(S)=N_{1}(S) \cap (M_{1}(B))$ for every $S \subseteq M_{2}(B)$. In short, we write $f(w)$ instead of
$f(\{w\})$.

\begin{lemma}
\label{fwclique} Let $w \in M_{2}(B)$. Then $G[f(w)]$ is
a clique.
\end{lemma}

\begin{proof}
There exists $b\in B$ such that $w\in M_{2}(b)$. It should be proved that
$N_{1}(w)\cap M_{1}(b)$ is a clique. Assume, on the contrary, that there exist
two nonadjacent vertices, $v_{1}$ and $v_{2}$, in $N_{1}(w)\cap M_{1}(b)$.
Then $G[\{b,v_{1},w,v_{2}\}]$ is isomorphic to $C_{4}$, which is a contradiction.
\end{proof}

\begin{lemma}
\label{m2p2} Let $w_{1}$ and $w_{2}$ be two adjacent vertices in $M_{2}(B_{X})
\cup M_{2}(B_{Y})$. Then at least one of the following holds:

\begin{itemize}
\item $f(w_{1}) \subseteq f(w_{2})$.

\item $f(w_{2}) \subseteq f(w_{1})$.
\end{itemize}
\end{lemma}

\begin{proof}
There exists $b\in B$ such that $\{w_{1},w_{2}\}\subseteq M_{2}(b)$. It should
be proved that at least one of the following inclusions holds:

\begin{itemize}
\item $N_{1}(w_{1})\cap M_{1}(b)\subseteq N_{1}(w_{2})\cap M_{1}(b)$.

\item $N_{1}(w_{2})\cap M_{1}(b)\subseteq N_{1}(w_{1})\cap M_{1}(b)$.
\end{itemize}

Assume, on the contrary, that there exist $v_{1}\in(N_{1}(w_{1})\setminus
N(w_{2}))\cap M_{1}(b)$ and $v_{2}\in(N_{1}(w_{2})\setminus N_{1}(w_{1}))\cap
M_{1}(b)$. If $v_{1}v_{2}\in E$ then $G[\{v_{1},v_{2},w_{2},w_{1}\}]$ is
isomorphic to $C_{4}$. Otherwise, $G[\{v_{1},b,v_{2},w_{2},w_{1}\}]$ is
isomorphic to $C_{5}$. In both cases we obtained a contradiction, which
completes the proof.
\end{proof}

\begin{lemma}
\label{fcclique} Let $C$ be a connected component of $G[M_{2}(B)]$. Then $f\left(  V\left(  C\right)  \right)  $ is a clique.
\end{lemma}

\begin{proof}
There exists $b\in B$ such that $C\subseteq M_{2}(b)$. Assume on the contrary
that $f\left(  V\left(  C\right)  \right)  $ is not a clique. Then there exist
two nonadjacent vertices, $v_{1}$ and $v_{2}$, in $f\left(  V\left(  C\right)
\right)  $.

If there existed a vertex $w\in N_{1}(v_{1})\cap N_{1}(v_{2})\cap M_{2}(b)$,
then $\{v_{1},v_{2}\}\subseteq f(w)$, which is a contradiction to Lemma
\ref{fwclique}. Hence, $N_{1}(v_{1})\cap N_{1}(v_{2})\cap M_{2}%
(b)=\emptyset$.

Let $P$ be a shortest path in $C$ between $N_{1}(v_{1})\cap M_{2}(b)$ and
$N_{1}(v_{2})\cap M_{2}(b)$. It holds that $G[V(P)\cup\{v_{1},b,v_{2}\}]$ is
an induced cycle of length $3+\left\vert V(P)\right\vert \geq5$, which
contradicts the fact that $G$ is chordal. Therefore, $f\left(  V\left(
C\right)  \right)  $ is a clique.
\end{proof}

\begin{lemma}
\label{monotonepath} Assume \ that \ $C$ is a connected component of
\ $G[M_{2}(B)]$, \ \ $w_{1},w_{2}\in V\left(  C\right)  $, and
$P=(w_{1}=u_{1},...,u_{r}=w_{2})$ is a shortest path in $C$ between $w_{1}$
and $w_{2}$.

\emph{(i)} Suppose there exists a vertex $v\in f(w_{1})\setminus f(w_{2})$.
Then there exists an index $1\leq s<r$ such that $v\in f(u_{i})\iff i\leq s$.

\emph{(ii)} Suppose there exists a vertex $v\in f(w_{2})\setminus f(w_{1})$.
Then there exists an index $1<t\leq r$ such that $v\in f(u_{i})\iff i\geq t$.
\end{lemma}

\begin{proof}
\emph{(i) }Define

\[
s=min\{i:1\leq i<r,v\in f(u_{i}),v\not \in f(u_{i+1})\},
\]

and assume on the contrary that there exists $s+2\leq s^{\prime}\leq r$ such
that $v\in f(u_{s^{\prime}})$. Then $G[v,u_{s},...,u_{s^{\prime}}]$ is an
induced cycle of length $s^{\prime}-s+2\geq4$, which contradicts the fact that
the graph is chordal.

\emph{(ii) }Similar to Case \emph{(i)}.
\end{proof}

\begin{lemma}
\label{maxc} Let $C$ be a connected component of $G[M_{2}(B)]$, and let $w_{1}$ and $w_{2}$ be two vertices in $C$. Then there
exists a vertex $w\in V\left(  C\right)  $ such that $f(w_{1})\cup
f(w_{2})\subseteq f(w)$.
\end{lemma}

\begin{proof}
If $f(w_{1})\subseteq f(w_{2})$ then $w=w_{2}$. Similarly, if $f(w_{2}%
)\subseteq f(w_{1})$ then $w=w_{1}$. Hence, assume that $f(w_{1})\setminus
f(w_{2})\neq\emptyset$ and $f(w_{2})\setminus f(w_{1})\neq\emptyset$.

Let $P=(w_{1}=u_{1},...,u_{r}=w_{2})$ be a shortest path in $C$ between
$w_{1}$ and $w_{2}$. Define

\[ s=\max\{i:1\leq i\leq r,f(u_{i})\supseteq
f(w_{1})\}; t=\min\{i:s\leq i\leq r,f(u_{i})\supseteq f(w_{2})\}.
\]

If $s=t$ then $f(w_{1})\cup f(w_{2})\subseteq f(u_{s})$. Therefore, one may
assume that $s<t$.

There exist $v_{1}\in f(u_{s})\setminus f(u_{s+1})$ and $v_{2}\in
f(u_{t})\setminus f(u_{t-1})$. By Lemma \ref{monotonepath}, $v_{1}$ is not
adjacent to $u_{s+1},...,u_{r}$ and $v_{2}$ is not adjacent to $u_{1}%
,...,u_{t-1}$. By Lemma \ref{fcclique}, $v_{1}v_{2}\in E\left(  G\right)  $.
Hence, $G[v_{1},u_{s},...,u_{t},v_{2}]$ is an induced cycle of size
$t-s+3\geq4$, which is a contradiction.
\end{proof}

\begin{lemma}
\label{maxcw} Let $C$ be a connected component of $G[M_{2}(B)]$. Then there exists a vertex $w\in V\left(  C\right)  $ such
that $f(w)=f(V\left(  C\right)  )$.
\end{lemma}

\begin{proof}
Follows immediately from Lemma \ref{maxc}.
\end{proof}

\newpage

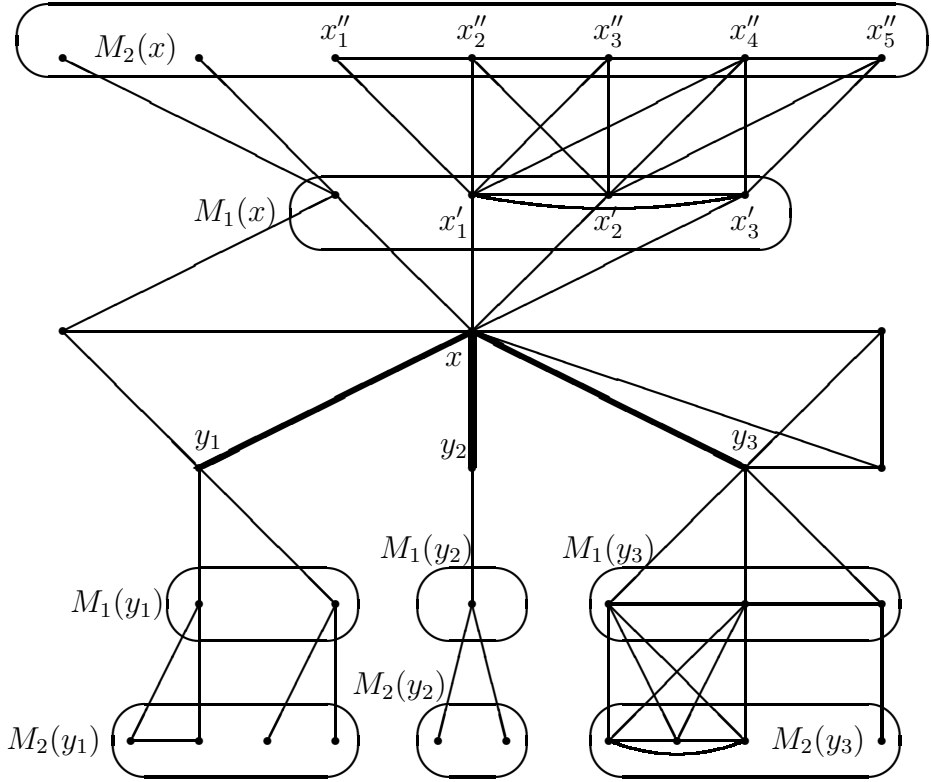
\begin{figure}[h]
\setlength{\unitlength}{1.21cm}

\begin{picture}(-5,9)\thicklines
\put(5,5){\circle*{0.1}}
\put(4.8,4.7){\makebox(0,0){$x$}}
\multiput(2,3.5)(3,0){3}{\circle*{0.1}}
\put(2.1,3.8){\makebox(0,0){$y_{1}$}}
\put(4.8,3.7){\makebox(0,0){$y_{2}$}}
\put(8,3.8){\makebox(0,0){$y_{3}$}}
\multiput(4.95,5)(0.02,0){6}{\line(2,-1){3}}
\multiput(4.97,5)(0.02,0){4}{\line(0,-1){1.5}}
\multiput(4.95,5)(0.02,0){6}{\line(-2,-1){3}}
\multiput(5,6.5)(1.5,0){3}{\circle*{0.1}}
\put(5,5){\line(2,1){3}}
\put(5,5){\line(1,1){1.5}}
\put(5,5){\line(0,1){1.5}}
\multiput(3.5,8)(1.5,0){5}{\circle*{0.1}}
\qbezier(5,6.5)(6.5,6.2)(8,6.5)
\put(5,6.5){\line(1,0){3}}
\put(5,6.5){\line(-1,1){1.5}}
\put(5,6.5){\line(0,1){1.5}}
\put(5,6.5){\line(1,1){1.5}}
\put(5,6.5){\line(2,1){3}}
\put(6.5,6.5){\line(-1,1){1.5}}
\put(6.5,6.5){\line(0,1){1.5}}
\put(6.5,6.5){\line(1,1){1.5}}
\put(6.5,6.5){\line(2,1){3}}
\put(8,6.5){\line(0,1){1.5}}
\put(8,6.5){\line(1,1){1.5}}
\put(3.5,8){\line(1,0){6}}
\put(3.5,8.3){\makebox(0,0){$x_{1}''$}}
\put(5,8.3){\makebox(0,0){$x_{2}''$}}
\put(6.5,8.3){\makebox(0,0){$x_{3}''$}}
\put(8,8.3){\makebox(0,0){$x_{4}''$}}
\put(9.5,8.3){\makebox(0,0){$x_{5}''$}}
\put(3.5,6.5){\circle*{0.1}}
\put(2,8){\circle*{0.1}}
\put(0.5,8){\circle*{0.1}}
\put(5,5){\line(-1,1){3}}
\put(3.5,6.5){\line(-2,1){3}}
\put(4.8,6.2){\makebox(0,0){$x_{1}'$}}
\put(6.5,6.2){\makebox(0,0){$x_{2}'$}}
\put(8,6.2){\makebox(0,0){$x_{3}'$}}
\put(5,8.2){\oval(10,0.8)}
\put(5.75,6.3){\oval(5.5,0.8)}
\put(1.3,8.1){\makebox(0,0){$M_{2}(x)$}}
\put(2.4,6.3){\makebox(0,0){$M_{1}(x)$}}
\multiput(2,2)(1.5,0){6}{\circle*{0.1}}
\put(2,3.5){\line(0,-1){1.5}}
\put(2,3.5){\line(1,-1){1.5}}
\put(5,3.5){\line(0,-1){1.5}}
\put(8,3.5){\line(-1,-1){1.5}}
\put(8,3.5){\line(0,-1){1.5}}
\put(8,3.5){\line(1,-1){1.5}}
\multiput(1.25,0.5)(0.75,0){4}{\circle*{0.1}}
\multiput(4.625,0.5)(0.75,0){2}{\circle*{0.1}}
\multiput(6.5,0.5)(0.75,0){3}{\circle*{0.1}}
\put(9.5,0.5){\circle*{0.1}}
\put(2,2){\line(-1,-2){0.75}}
\put(2,2){\line(0,-1){1.5}}
\put(3.5,2){\line(-1,-2){0.75}}
\put(3.5,2){\line(0,-1){1.5}}
\put(5,2){\line(-1,-4){0.375}}
\put(5,2){\line(1,-4){0.375}}
\put(6.5,2){\line(1,-1){1.5}}
\put(6.5,2){\line(1,-2){0.75}}
\put(6.5,2){\line(0,-1){1.5}}
\put(8,2){\line(-1,-1){1.5}}
\put(8,2){\line(-1,-2){0.75}}
\put(8,2){\line(0,-1){1.5}}
\put(9.5,2){\line(0,-1){1.5}}
\put(1.25,0.5){\line(1,0){0.75}}
\put(6.5,0.5){\line(1,0){1.5}}
\qbezier(6.5,0.5)(7.25,0.2)(8,0.5)
\put(6.5,2){\line(1,0){3}}
\put(2.7,2){\oval(2.1,0.8)}
\put(5,2){\oval(1.2,0.8)}
\put(8,2){\oval(3.4,0.8)}
\put(2.4,0.5){\oval(2.7,0.8)}
\put(5,0.5){\oval(1.2,0.8)}
\put(8,0.5){\oval(3.4,0.8)}
\put(1.1,2){\makebox(0,0){$M_{1}(y_{1})$}}
\put(6.5,2.6){\makebox(0,0){$M_{1}(y_{3})$}}
\put(0.4,0.5){\makebox(0,0){$M_{2}(y_{1})$}}
\put(8.8,0.5){\makebox(0,0){$M_{2}(y_{3})$}}
\put(4.5,2.6){\makebox(0,0){$M_{1}(y_{2})$}}
\put(4.2,1.1){\makebox(0,0){$M_{2}(y_{2})$}}
\put(9.5,5){\circle*{0.1}}
\put(9.5,3.5){\circle*{0.1}}
\put(9.5,5){\line(-1,0){4.5}}
\put(9.5,3.5){\line(-1,0){1.5}}
\put(9.5,3.5){\line(0,1){1.5}}
\put(9.5,3.5){\line(-3,1){4.5}}
\put(9.5,5){\line(-1,-1){1.5}}
\put(0.5,5){\circle*{0.1}}
\put(0.5,5){\line(1,0){4.5}}
\put(0.5,5){\line(1,-1){1.5}}
\put(0.5,5){\line(2,1){3}}
\end{picture}

\caption{$B = G[\{x, y_{1},y_{2},y_{3}\}]$ is a generating subgraph. By Lemmas
\ref{mbx} and \ref{myi}, $M_{1}(x)$, $M_{1}(y_{1})$, $M_{1}(y_{2})$, $M_{1}(y_{3})$ are
mutually disjoint and nonadjacent. By Lemmas \ref{m2bx} and \ref{m2yi}, also
$M_{2}(x)$, $M_{2}(y_{1})$, $M_{2}(y_{2})$, $M_{2}(y_{3})$ are mutually
disjoint and nonadjacent. Moreover, $C = G[\{x_{1}^{\prime\prime}%
,...,x_{5}^{\prime\prime}\}]$ is a connected component of $M_{2}(x)$, and
$(x_{1}^{\prime\prime},...,x_{5}^{\prime\prime})$ is a shortest path between
$x_{1}^{\prime\prime}$ and $x_{5}^{\prime\prime}$ in $C$. $f(x_{1}%
^{\prime\prime})=\{x_{1}^{\prime}\}$, $f(x_{2}^{\prime\prime})=\{x_{1}%
^{\prime}, x_{2}^{\prime}\}$, $f(x_{3}^{\prime\prime})=\{x_{1}^{\prime}%
,x_{2}^{\prime}\}$, $f(x_{4}^{\prime\prime})=\{x_{1}^{\prime}, x_{2}^{\prime},
x_{3}^{\prime}\}$, $f(x_{5}^{\prime\prime})=\{x_{2}^{\prime}, x_{3}^{\prime
}\}$. By Lemma \ref{fcclique}, $f(C)$ is a clique. By Lemma \ref{maxcw},
$f(C)=f(x_{4}^{\prime\prime})$. Note that there are vertices which are
adjacent to both $B_{X}$ and $B_{Y}$, but they are not important for the
algorithm.}%
\label{conj}
\end{figure}

\vspace{0.4cm}

The following algorithm receives as its input a chordal graph $G$ and an
induced complete bipartite subgraph $B$ of $G$. The algorithm decides whether
$B$ is generating.

\newpage

\begin{algorithm}
\DontPrintSemicolon
{ {\bf find} $M(B)$.  \; }
{ {\bf find} $M_{2}(B)$.  \; }
{ {\bf find} the connected components of $M_{2}(B)$.  \; }
{ \ForEach{$w \in M_{2}(B)$ }
{ {\bf calculate} $|f(w)|$. }
}
{ $S \leftarrow \emptyset$.  \; }
{ \ForEach{connected component $C$ of $M_{2}(B)$ }
{
{ {\bf find} a vertex $w_{C} \in C$ such that $|f(w_{C})|$ is maximal. \; }
{ $S \leftarrow S \cup \{w_{C}\}$.  \; }
}
}
{ \If{$S$ dominates $M(B)$}
{ {\bf output} ``$B$ is generating". }
}
{ \Else
{
{ {\bf output} ``$B$ is not generating". }
}
}
\caption{Recognizing generating subgraphs in chordal graphs\label{genalg}}
\end{algorithm}

\textbf{Correctness of Algorithm \ref{genalg}:} The set $S$ is independent,
because it contains one vertex from each connected component of $M_{2}(B)$. Let $S^{\prime}$ be another independent set of $M_{2}(B)$. We prove that every vertex in $M_{1}(B)$ which is dominated by $S^{\prime}$ is also 
dominated
by $S$. Assume on the contrary that there exists a vertex $v\in M_{1}(B)$ which is dominated by $S^{\prime}$, but not by $S$.
Let $w^{\prime}\in S^{\prime}\cap N(v)$, let $C$ be the connected component of $M_{2}(B)$ which contains $w^{\prime}$, and let $w$ be
the vertex in $C$ which belongs to $S$. It follows from the construction of
$S$ and Lemma \ref{maxcw} that $f(w^{\prime})\subseteq f(w)$, which is a
contradiction. Therefore, $N_{1}(S^{\prime}) \cap M_{1}(B) \subseteq N_{1}(S) \cap M_{1}(B)$.

If $S$ dominates $M_{1}(B)$, then let $S^{\ast}$ be any
maximal independent set of $G[V(G)\setminus N_{1}[B]]$ which contains $S$.
Clearly, $S^{\ast}$ is a witness that $B$ is generating. However, if $S$ does
not dominate $M_{1}(B)$, then there does not exist an
independent set in $M_{2}(B)$ which dominates $M_{1}(B)$, and therefore, $B$ is not generating.

\textbf{Complexity of Algorithm \ref{genalg}:} Each stage of the algorithm can
be implemented in $O(|V|^{2})$ time. Therefore, this goes in parallel with the
time complexity of the whole algorithm.

\begin{corollary}
\label{rechordal} Recognizing relating edges and generating subgraphs in
chordal graphs can be done polynomially.
\end{corollary}

\section{Conclusions and future work}

In \cite{LevTan2016} the following four problems have been defined.

\begin{itemize}
\item $\mathbf{WC}$\textbf{ problem}:\newline\textit{Input}: A graph
$G$.\newline\textit{Question}: Is $G$ well-covered?

\item $\mathbf{WCW}$\textbf{ problem}:\newline\textit{Input}: A graph
$G$.\newline\textit{Output}: The vector space $WCW(G)$.

\item $\mathbf{GS}$ \textbf{problem}:\newline\textit{Input}: A graph $G$, and
an induced complete bipartite subgraph $B$ of $G$.\newline\textit{Question}:
Is $B$ generating?

\item $\mathbf{RE}$\textbf{ problem}:\newline\textit{Input}: A graph $G$, and
an edge $xy\in E\left(  G\right)  $.\newline\textit{Question}: Is $xy$
relating?\newline
\end{itemize}

It concluded in a table presenting complexity results on the above four
problems for various graphs. The findings of the current paper may be
considered as an extra line for this table. Specifically, every entry of this
line claims that the corresponding problem is polynomial for chordal graphs.
It was proved that the $\mathbf{GS}$ problem and the $\mathbf{RE}$ problem are
\textbf{NPC}\ even for bipartite graphs \cite{LevTan2016}. On the other hand, for this family of graphs,
it is known that the $\mathbf{WC}$\textbf{ }problem is polynomial
\cite{ravindra:well-covered}, while the complexity status of the
$\mathbf{WCW}$\textbf{ }problem is still open.

Since both
chordal graphs and bipartite graphs are perfect, it seems natural to
investigate perfect graphs with polynomially solvable $\mathbf{WC}$ and/or
$\mathbf{WCW}$ problems. Some of such subclasses of graphs are known. For instance, those with bounded
clique size and those with no induced $C_{4}$ \cite{dz:wcge}.

\newpage
\textbf {References}

\bibliographystyle{elsarticle-num}


\end{document}